\documentclass{article}
\usepackage[utf8]{inputenc}
\usepackage{fancyhdr}
\usepackage{amsmath}
\usepackage{esint}
\usepackage{enumerate}
\usepackage{mathtools}
\usepackage{amsmath,amssymb}
\usepackage{dsfont}
\usepackage[margin=1.0in]{geometry}
\usepackage{graphicx}
\usepackage[mathscr]{euscript}
\usepackage{tikz}
\usepackage{amssymb}
\usepackage{mathrsfs}
\usepackage{pgfplots}
\usepackage{caption}
\usepackage{cleveref}
\usepackage{epstopdf}
\usepackage{float}
\usetikzlibrary{calc,trees,positioning,arrows,fit,shapes,calc}
\usepackage{listings}
\usepackage{amsthm}
\graphicspath{ {Uni/} }
\usepackage[space]{grffile}
\usepackage[citestyle=numeric,bibstyle=numeric,maxbibnames=99,
uniquename=init,autolang=hyphen,backend=biber,sortcites,
sorting=nyt]{biblatex}
\theoremstyle{plain}
\newtheorem{definition}{Definition}
\newtheorem{theorem}[definition]{Theorem}
\newtheorem{remark}[definition]{Remark}
\newtheorem{corollary}[definition]{Corollary}

\DeclarePairedDelimiterX\innerp[2]{\langle}{\rangle}{
	#1, #2
}
\DeclareMathOperator*{\Tr}{Tr}

\begin{document}
	\author{Lukas Junge}
	\title{Propagation of Condensation via Neumann Localization in the Dilute Bose Gas}
	\date{\today}
	\maketitle
\begin{abstract}
	We prove a Neumann localization inequality for the Laplacian that includes a spectral gap. This result is obtained by partitioning a cube into overlapping families of subcubes and analysing the associated projection operators. The resulting operator inequality goes through a discrete Neumann Laplacian on the lattice of boxes and yields a quantitative spectral gap estimate.
	
	As an application, we consider the dilute Bose gas with Neumann boundary conditions. Combining the localization method with recently established free-energy bounds \cite{haberberger2024freeenergydilutebose}, we propagate strong condensation estimates obtained in \cite{freeenergydilutebose} from the Gross–Pitaevskii scale to larger boxes of side length $R\sim a(\rho a^3)^{-\frac{3}{4}-\eta}$.
\end{abstract}	
\section*{Introduction}
Since the birth of mathematical physics the rigorous derivation of Bose–Einstein condensation (BEC) from the microscopic many-body Hamiltonian has remained a central problem. While the ground state energy of the dilute Bose gas is being derived to higher and higher precision, establishing condensation directly from first principles — particularly at positive temperature and on physically relevant length scales — is substantially more delicate. Although the thermodynamic limit is still not within reach, there have been significant contributions in extending the scale at which we can prove BEC, see for instance \cite{LScondensation,fournais2020lengthscalesbecdilute,AnotherGPBEC,chong2025kineticlocalizationpoincaretypeinequalities}.

Recent advances have provided precise control of the free energy and established strong condensation estimates in Neumann boxes whose side lengths are slightly larger than the Gross–Pitaevskii scale; see in particular \cite{haberberger2024freeenergydilutebose} and an extension in \cite{freeenergydilutebose}. These results yield quantitative bounds on the number of excited particles, but only on relatively short spatial scales.

The purpose of this paper is to extend such condensation estimates to longer length scales similar to what was done in the periodic setting in \cite{fournais2020lengthscalesbecdilute} based upon methods of \cite{FS,BSol}. Contrary to \cite{fournais2020lengthscalesbecdilute} this is done in the Neumann setting, which is a non trivial variation. This technique also leads to a kinetic operator significantly simpler than what was used in \cite{FS,FS2}. 
In \cite{chong2025kineticlocalizationpoincaretypeinequalities} the authors did something similar, however they needed to sacrifice a fraction of the kinetic energy which had significant cost to their condensation estimates, this we avoid completely here.

The key technical input is a localization technique: by partitioning the domain into overlapping families of boxes and analysing the associated projection operators, we derive a lower bound that uses a discrete Neumann Laplacian with vertices being the boxes and edges connecting neighbouring boxes. Spectral analysis of this operator yields an estimate, which can be iterated across scales. The result is optimal in the scale parameter. 

This Neumann localization allows us to propagate the strong condensation estimates obtained on the Gross–Pitaevskii scale \cite{freeenergydilutebose} to weak condensation estimates on significantly larger boxes. In particular, we show that condensation persists on scales $R\sim a(\rho a^3)^{-3/4-\eta}$ and at temperatures $T\sim \rho a$. Thereby extending the regime in which BEC can be rigorously established. This is done in the final section.

\section*{The Neumann localisation}
We consider the Neumann Laplacian $-\Delta_{\Lambda_R}$ on a box $\Lambda_R=[0,R]^3$. Its lowest eigenvalue is $0$ with corresponding eigenvector the constant function $1_{\Lambda_R}$. Let $P$ denote the projection onto the supspace spanned by $1_{\Lambda_R}$ and $Q=1-P$. For any subset $A\subset \Lambda_R$, we denote $P_A$ the projection onto constant functions on $A$ and $Q_A=1_A-P_A$. Our goal is to subdivide $\Lambda_R=\sqcup_i A_i$  such that
\begin{equation}\label{what we really want}
	-\Delta_{\Lambda_R}-\frac{Q}{R^2}\geq \sum_{i}\left(-\Delta_{A_i}-c\lambda_1(A_i)Q_{A_i}\right),
\end{equation}
with $\lambda_1$ being the first non-zero eigenvalue of $-\Delta_{A_i}$ and $c>0$ a universal constant.
In a quadratic form sense, we have the operator inequality 
\[-\Delta_{\Lambda_R}\geq \sum_{i}-\Delta_{A_i}.\]
Thus in order to establish \cref{what we really want} it would suffice to show that
\begin{equation*}
\frac{Q}{R^2}\leq \sum_{i}c\lambda_1(A_i)Q_{A_i}.
\end{equation*}
However, this inequality is false in general, since all the operators $Q_{A_i}$ may vanish simultaneously even when the function is not globally constant.To overcome this issue, we introduce multiple overlapping partitions of $\Lambda_R$. The overlaps ensure that one can instead prove an inequality of the form
\begin{equation}\label{some inequality}
	\frac{Q}{R^2}\leq \sum_{k}\sum_{i}c\lambda_1(A_{i,k})Q_{A_{i,k}}.
\end{equation}
In what follows, we will restrict our attention to the case where the sets $A$ are cubes, as this is the only shape directly applicable to the Bose gas. Since the target inequality is scale-invariant, it suffices to establish it for $R=1$

Theorem~\ref{theorem easy} uses two partitions in an elegant and optimal way. However the boundary sets of one of the partitions are not cubes and the theorem is therefore not directly applicable to the Bose gas. In Theorem~\ref{theorem hard} we consider many partitions consisting only of cubes. While the underlying idea remains the same, the additional technicalities make the argument somewhat less transparent.
\begin{theorem}\label{theorem easy}
	For $\ell^{-1}\in \mathbb{N}$, we consider the two partitions
	\begin{equation}
		\bigcup_{k\in \{0,1,..,\ell^{-1}-1\}^3}B_{\ell,k}=[0,1]^3=	\bigcup_{k\in \{0,1,..,\ell^{-1}-1\}^3}B_{\ell,k}^{shift}.
	\end{equation}
where 
\begin{equation*}
	B_{\ell,k}=\prod_{i=1} ^3[k_i\ell,(k_i+1)\ell],\qquad B_{\ell,k}^{shift}=\prod_{i=1} ^3[(k_i+\frac{1}{2}1_{k_i>  0})\ell,(k_i+1+\frac{1}{2}1_{k_i<\ell^{-1}-1})\ell].
\end{equation*}
Then for all $\ell^{-1}\in \mathbb{N}$, the following operator inequality holds on $L^2([0,1]^3)$:
\begin{equation}\label{only relevant equation}
\sum_{k}Q_{B_{\ell,k}}+Q_{B_{\ell,k}^{shift}}\geq \frac{1-\cos\pi \ell}{4+1-\cos\pi \ell}Q.
\end{equation}
Note that the sets $B_{\ell,k}^{shift}$ are not all cubes: intervals adjacent to the boundary at $1$ have length $\frac{\ell}{2}$, while those adjacent to the boundary at $0$ have length $\frac{3\ell}{2}$.
\end{theorem}
\begin{remark}
	The constant appearing in \cref{only relevant equation} is optimal in the limit $\ell \rightarrow 0$. Indeed by defining the function
	\begin{equation}
u=\sum_{s=0}^{2\ell^{-1}-1}cos(\frac{\pi \ell}{2}(s+\frac{1}{2}))1_{s\frac{\ell}{2}\leq x_1\leq (s+1)\frac{\ell}{2}},
	\end{equation}
and evaluating \cref{only relevant equation} on $u$ yields
\begin{equation}
	\langle u, \sum_{k}\left(Q_{B_{\ell,k}}+Q_{B_{\ell,k}^{shift}} \right)u\rangle\geq (\frac{\pi^2\ell^2}{8}- O(\ell^3))\Vert Qu\Vert^2.
\end{equation}
Which establishes the optimality of \cref{only relevant equation}.
\end{remark}
\begin{proof}
Since both sides of \cref{only relevant equation} vanish on constant functions, it suffices to verify the inequality on the orthogonal complement of constant functions. Because the boxes $B_{\ell,k}$ are disjoint, up to measure zero, the operators
\begin{equation}\label{definition of Qell 1}
	Q_{\ell}:=\sum_kQ_{B_{\ell,k}},\qquad Q_{\ell}^{shift}:=\sum_kQ_{B_{\ell,k}^{shift}}
\end{equation} 
are projections. We define $P_{\ell}$ and $P_{\ell}^{shift}$ analogously, naturally $Q_{\ell}=1-P_{\ell}$. By Cauchy-Schwartz inequality we find
\begin{equation}\label{writing Q+Q as Q-Q squared}
P_{\ell}Q_{\ell}^{shift}P_{\ell}=Q_{\ell}^{shift}-Q_{\ell}Q_{\ell}^{shift}-Q_{\ell}^{shift}Q_{\ell}+Q_{\ell}Q_{\ell}^{shift}Q_{\ell}\leq 2(Q_{\ell}+Q_{\ell}^{shift}).
\end{equation}
Thus we are reduced to studying the operator
\begin{equation}\label{definition of L}
	L(\ell)=P_{\ell}Q_{\ell}^{shift}P_{\ell}=P_{\ell}-P_{\ell}P_{\ell}^{shift}P_{\ell}.
\end{equation}
The operator $P_{\ell}$ maps $L^2([0,1]^3)$ onto the subspace of functions that are constant on each box $B_{\ell,k}$. Explicitly,
\begin{equation*}
	P_{\ell}(\psi)=\sum_k\left(\frac{1}{\ell ^3}\int_{B_{\ell,k}}\psi dx \right)1_{B_{\ell,k}}
\end{equation*}
This yields a natural isometry between $P_{\ell}L^2([0,1]^3)$ and $\ell^{2}(\{0,1,..,\ell^{-1}-1\}^3)$, given by
\begin{equation*}
	(G\psi)(k)=\frac{1}{\ell^3}\int_{B_{\ell,k}}\psi dx,\qquad k\in \{0,1,..,\ell^{-1}-1\}^3.
\end{equation*}
 Through this isometry we view $L(\ell)$ as an operator on $\ell^{2}(\{0,1,..,\ell^{-1}-1\}^3)$. A direct computation yields
\begin{align*}
	(L(\ell) f)(k)&=\sum_{s}\sum_{s'} \frac{\vert B_{\ell,k}\cap B_{\ell,s'}^{shift}\vert \vert B_{\ell,s}\cap B_{\ell,s'}^{shift}\vert}{\ell^3\vert B_{\ell,s'}^{shift}\vert} (f(k)-f(s))
\end{align*}
Evaluating the overlaps explicitly, one finds for $k$ away from the boundary,
\begin{align}\label{main equation}
	(L(\ell) f)(k)&=\frac{1}{16}\left(\sum_{\vert k-s\vert_1= 1} f(k)-f(s)+\frac{1}{2}\sum_{\substack{\vert k-s\vert_{\infty}=1\\ \vert k-s\vert_1= 2}}f(k)-f(s)+\frac{1}{4}\sum_{\substack{\vert k-s\vert_{\infty}=1\\ \vert k-s\vert_1= 3}}f(k)-f(s)\right).
\end{align}
One can easily check that the edges connecting to the boundary have higher weights. Therefore,  as operators on $\ell^2(\{0,1,..,\ell^{-1}-1\}^3)$, we find
\begin{equation}
	L(\ell)\geq -\Delta^{Graph}
\end{equation}
with $-\Delta^{Graph}$ being $1/16$ times the free/Neumann Laplacian on the cubic lattice, with edges of weight $1$ between nearest neighbours, edges of weight $\frac{1}{2}$ between "face diagonal" vertices, and edges of weight $\frac{1}{4}$ between "body diagonal" vertices. $-\Delta^{Graph}$  is diagonal in the discrete Neumann basis. More precisely,
\begin{equation}\label{diagonal laplacian}
	\begin{aligned}
		&-\Delta^{Graph}\phi_m=\lambda_m\phi_m,\quad \text{with}\quad \phi_m(k)=\prod_{i=1}^3\cos({m_i\pi\ell}(k_i+\frac{1}{2})),\qquad m\in \{0,1,2,..,\ell^{-1}-1\}^3\\& \lambda_m=\frac{1}{8}\sum_i(1-\cos({\pi m_i\ell}))+\frac{1}{8}\sum_{i<j}(1-\cos({\pi m_i\ell})\cos({\pi m_j\ell})+\frac{1}{8}\left(1-\prod_{i}\cos({\pi m_i\ell})\right).
	\end{aligned}
\end{equation}
We find on the space orthogonal to constant functions (i.e excluding $\phi_{(0,0,0)}$) the lowest eigenvalue is attained at $m=(1,0,0)$, and we obtain
\begin{equation*}
	L(\ell)\geq -\Delta^{graph}\geq  \lambda_{(1,0,0)}P_{\ell}=\frac{(1-\cos\pi\ell)}{2}P_{\ell}
\end{equation*}
The above together with \cref{definition of L} and \cref{writing Q+Q as Q-Q squared} yields
\begin{equation*}
	Q_\ell+Q_{\ell}^{shift}\geq \frac{(1-\cos\pi\ell)}{4}P_{\ell}=\frac{(1-\cos\pi\ell)}{4}(1-Q_{\ell})\implies Q_\ell+Q_{\ell}^{shift}\geq \frac{1-\cos\pi \ell}{4+1-\cos \pi \ell}
\end{equation*}
on functions orthogonal to constant functions. 
\end{proof}

Since the results of \cite{freeenergydilutebose} have only been proven for cubes the previous theorem is not directly applicable. The following Theorem remedies this by considering families of partitions consisting exclusively of cubes of varying sizes. By varying the side lengths, we obtain sufficient overlap to recover an inequality of the same form. Although the construction is more technical, it allows for a direct application to the dilute Bose gas without further modifications.

\begin{theorem}\label{theorem hard}
	Given $\ell^{-1}\in \mathbb{N}$ we define
	\begin{equation}\label{formula for ell n}
		\ell_n=\frac{\ell}{1-n\ell},\qquad 0\leq n\leq \frac{1}{2}\ell^{-1}
	\end{equation}
That is, the values $\ell_n$'s are all the inverse integers between $\ell$ and $2\ell$. For each $\ell_n$ we consider the partition into cubes
\begin{equation*}
		\bigcup_{k\in \{0,1,..,\ell_n^{-1}-1\}^3}B_{\ell_n,k}=[0,1]^3,\qquad\qquad 	B_{\ell_n,k}=\prod_{i=1}^3[k_i\ell_n,(k_i+1)\ell_n].
\end{equation*}
We also define
\begin{equation}\label{definion of Qell}
	Q_{\ell_n}=\sum_{ k\in \{0,1,..,\ell_n^{-1}-1\}^3}Q_{B_{\ell_n,k}}
\end{equation}
Then, there exist a constant $C>0$ independent of $\ell$, such that the following operator inequality holds on $L^2([0,1]^3)$:
\begin{equation}\label{nice inequality}
\ell\sum_{n=0}^{\lfloor\frac{1}{2\ell}\rfloor}Q_{\ell_n}\geq C\ell^2Q.
\end{equation}

\end{theorem}
\begin{proof}
	Using the same inequality as in \cref{writing Q+Q as Q-Q squared} we have for each pair $\ell_n$, $\ell_m$
	\begin{equation*}
		P_{\ell_n}Q_{\ell_m}P_{\ell_n}\leq 2(Q_{\ell_n}+Q_{\ell_m}),
	\end{equation*}
summing over $n$ and $m$ yields
	\begin{equation*}
	\sum_{n=0}^{\lfloor\frac{1}{2\ell}\rfloor}Q_{\ell_n}\geq \frac{C}{\ell}\sum_{m,n} P_{\ell_m}Q_{\ell_n}P_{\ell_m}.
	\end{equation*}
Thus, it suffices to prove that for each $m$;
\begin{equation}\label{second wanted equation}
L(\ell_m)=\sum_{n} P_{\ell_m}Q_{\ell_n}P_{\ell_m}\geq \frac{C}{\ell}P_{\ell_m}.
\end{equation}
Proceeding as in \cref{main equation} we represent $L(\ell_m)$ as the discrete operator:
\begin{equation}
(L(\ell_m)f)(k)=\sum_{n}\sum_{s}\sum_{s'}\frac{\vert B_{\ell_m,k}\cap B_{\ell_n,s'}\vert \vert B_{\ell_m,s}\cap B_{\ell_n,s'}\vert}{\ell_m^3\ell_n^3} (f(k)-f(s)).
\end{equation}
We bound $L(\ell_m)$ from below by restricting to nearest-neighbour edges, i.e. $\vert k-s\vert=1$, and simply bound the diagonal edges by $0$. Without loss of generality assume $k_1\neq \ell_m^{-1}-1$ and consider the weight of the edge between $k$ and $k+e_1$.
	\begin{align}
	&\sum_{n}\sum_{s'}\frac{\vert B_{\ell_m,k}\cap B_{\ell_n,s'}\vert \vert B_{\ell_m,k+e_1}\cap B_{\ell_n,s'}\vert}{\ell_m^3\ell_n^3}\nonumber\\&\geq C \sum_{n}\sum_{s_1'=0}^{\ell_n^{-1}-1} \frac{\vert [k_1\ell_m ,(\ell_m+1)k_1]\cap[s_1'\ell_n,(s_1'+1)\ell_n]\vert [(k_1+1)\ell_m ,(\ell_m+2)k_1]\cap[s_1'\ell_n,(s_1'+1)\ell_n]\vert}{\ell^2}. \nonumber
	\end{align}
Using the expression $\ell_n=\frac{\ell}{1-n\ell}$ the above expression can be bounded from below by
\begin{equation*}
C\sum_n\sum_{s_1'}\left(\frac{k+1}{1-m\ell}-\frac{s_1'}{1-n\ell}\right)_+\left(\frac{s_1'+1}{1-n\ell}-\frac{k+1}{1-m\ell}\right)_+.
\end{equation*}
The inner sum over $s_1'$ can be estimated by the following:
\begin{equation}
	\sum_{s_1'}\left(...\right)\geq 
	 Cd\left( \frac{(k+1)(1-n\ell)}{1-m\ell},\mathbb{Z}\right)=Cd\left(\frac{(k+1)(m-n)}{\ell^{-1}-m},\mathbb{Z}\right),
\end{equation}
where $d(\cdot, \mathbb{Z})$ is distance to nearest integer. Lastly since $\frac{k_1+1}{\ell^{-1}-m}$ is not an integer, and its distance to an integer is at least $\ell_m\geq \ell$ we may conclude
\begin{equation*}
	\sum_{n=0}^{\lfloor\frac{1}{2\ell}\rfloor}d\left(\frac{(k+1)(m-n)}{\ell^{-1}-m},\mathbb{Z}\right)\geq C\ell^{-1}.
\end{equation*}
It follows that $L(\ell_m)$ dominates the nearest neighbour discrete Laplacian with edge weight $C\ell^{-1}$. Using the same diagonalization as in \cref{diagonal laplacian} yields \cref{second wanted equation}. 
\end{proof}
\section*{Application to the dilute Bose gas}
We consider the dilute Bose gas described by the Hamiltonian
\begin{equation}\label{hamiltonnian}
	H_N=\sum_{i=1}^N-\Delta_i+\sum_{i<j}v(x_i-x_j),
\end{equation}
for a non negative potential $v$. $H_N$ acts on $L^2([0,L]^{3N})$ with Neumann boundary conditions. The gas is considered in the dilute regime where $\frac{N}{L^3}=\rho \ll a^{-3}$, with $a$ being the scattering length of $v$, see \cite{greenbook} for definitions of $a$. We study the free energy of the system and the occurrence of Bose–Einstein condensation in the corresponding Gibbs state. The free energy is defined by
\begin{equation}\label{free energy definition}
	\inf_{\Gamma\geq0, \Tr(\Gamma)=1} Tr(H_N\Gamma)+T\Tr(\Gamma\log(\Gamma))=-T\log \Tr(e^{-\beta H_N})
\end{equation}
where the equality follows by evaluating the functional at the Gibbs state 
\[\Gamma_0=\frac{1}{z}e^{-\beta H_N}.\]
To quantify condensation, we define $Q$ as the projection orthogonal to constant functions on $L^2([0,L]^3)$ and set
\begin{equation}
	n_+=\sum_{i=1}^NQ_i,
\end{equation}
where $Q_i$ acts on the $i$'th particle. We say $\Gamma_0$ exhibits condensation if
\begin{equation}
	\Tr(\Gamma_0 n_+)\leq (1-\varepsilon)N
\end{equation}
for some universal $\varepsilon>0$. This definition is consistent with standard notions of Bose–Einstein condensation; see, for instance, \cite{PhysRev.104.576}.

As mentioned in the introduction we are not able to prove BEC in the thermodynamic limit, instead we consider regimes in which $N$ and $\rho a^3$ are coupled by letting $L=af(\rho a^3)$.

In \cite{freeenergydilutebose} the following result was established.
\begin{theorem}\label{main bose theorem}
	Let $v\geq 0$ be radially decreasing and compactly supported, then there exist $C,\eta>0$ such that for $L= a (\rho a^3)^{-\frac{1}{2}-\eta}$ and $T\leq \rho a (\rho a^3)^{-\eta}$ we have
	\begin{align}\label{main bose equation}
		-T \log \Tr e^{-\beta (H_N-\eta\frac{n_+}{L^2})}&\geq 4\pi \rho a N(1+\frac{128}{15\sqrt{\pi}}\sqrt{\rho a^3}-C(\rho a^3)^{\frac{1}{2}+\eta})\nonumber\\&+L^3\frac{T^{\frac{5}{2}}}{(2\pi)^3}\int_{\mathbb{R}^3}\log (1-e^{-\sqrt{p^4+\frac{16\pi p^2 \rho a}{T}}})dp 
	\end{align}
\end{theorem}
\begin{remark}
At zero temperature, one can track the parameter $\eta$ explicitly in \cite{freeenergydilutebose}. The optimal value obtained there is
	\begin{equation}
		\eta= \frac{1}{32}
	\end{equation}
\end{remark}
From the variational principle \cref{free energy definition}, we find
\begin{equation*}
	\Tr(H_N\Gamma_0)+T\Tr(\Gamma_0\log \Gamma_0)\geq -T\log Tr(e^{-\beta(H_N-\eta \frac{n_+}{L^2})})+\Tr(\eta\frac{n_+}{L^2}\Gamma_0)
\end{equation*}
Combining this with Theorem~\ref{main bose theorem} and a corresponding upper bound on the free energy \cite{HABERBERGER2026110825}, we obtain
	\begin{equation}
	\Tr (n_+ \Gamma_0)\leq   CL^2N\rho a (\rho a^3)^{\frac{1}{2}+\eta}\leq   CN (\rho a^3)^{\frac{1}{2}-\eta}.
	\end{equation}
This yields a strong condensation estimate on boxes whose side length is slightly larger than the Gross–Pitaevskii scale, which corresponds to $L=a(\rho a^3)^{-\frac{1}{2}}$.

We now apply Theorem~\ref{theorem hard} to extend this strong condensation estimate to a weak condensation estimate on larger length scales. This was similarly done in \cite{fournais2020lengthscalesbecdilute} for the periodic setting. 
\begin{corollary}\label{main corollary}
	Let $v,T, \eta$ and $C$ be given as in Theorem~\ref{main bose theorem}. Then for $R\geq a (\rho a^3)^{-\frac{1}{2}-\eta}$ we have
\begin{equation}
	\frac{\Tr (n_+ e^{-\beta H_N})}{\Tr (e^{-\beta H_N})}\leq CN\rho R^2a (\rho a^3)^{\frac{1}{2}+\eta}.
\end{equation}
with $H_N$ given as in \cref{hamiltonnian} defined on $L^2([0,R]^{3N})$ satisfying $N=\rho R^3$. 
\end{corollary}
\begin{remark}
	In much of the Bose gas literature, one fixes the length $L=1$ and considers
	\begin{equation*}
		H_N(\kappa)=\sum_{i=1}^{N}-\Delta_i+\sum_{i<j}N^{2-2\kappa}v(N^{1-\kappa}(x_i-x_j))
	\end{equation*}
acting on $L^2([0,1]^{3N})$. In this parametrization, $\kappa=0$ corresponds to the Gross–Pitaevskii regime, while $\kappa=\frac{2}{3}$ corresponds to the thermodynamic scaling. Corollary~\ref{main corollary} implies condensation of the Gibbs state whenever
\[\kappa \in [0,\frac{2+2\eta}{5+3\eta}].\]
\end{remark}
\begin{proof}
Let $L=a (\rho a^3)^{-\frac{1}{2}-\eta}$ as in Theorem~\ref{main bose theorem}. Then according to Theorem~\ref{theorem hard} we may write
	\begin{equation*}
		-\Delta =\frac{1}{\lfloor\frac{L}{2R}\rfloor+1}\sum_{n=0}^{\lfloor\frac{L}{2R}\rfloor}-\Delta-\varepsilon \frac{Q}{R^2}+\varepsilon \frac{Q}{R^2}\geq\frac{1}{\lfloor\frac{L}{2R}\rfloor+1}\sum_{n=0}^{\lfloor\frac{L}{2R}\rfloor} -\Delta_{L_n}-\varepsilon\frac{CQ_{L_n}}{L^2}+\varepsilon \frac{Q}{R^2}
	\end{equation*}
where $-\Delta_{L_n}$ and $Q_{L_n}$ are given as in \cref{definion of Qell} and we used the inequality;
\begin{equation*}
	-\Delta\geq -\Delta_{L_n}= \sum_{k\in \{0,1,...,\frac{R}{L_n}-1\}^3}-\Delta_{B_{L_n,k}},\qquad L\leq  L_n\leq 2L.
\end{equation*}
Neglecting positive interactions between different boxes, we also obtain the lower bound
\begin{equation*}
v(x_i-x_j)\geq v_{L_n}(x_i-x_j)=\sum_{k\in \{0,1,...,\frac{R}{L_n}-1\}^3}1_{B_{L_n,k}}(x_i) v(x_i-x_j)1_{B_{L_n,k}}(x_j).
\end{equation*}
Combining these estimates we obtain the lower bound on $H_N$
\begin{equation}\label{bound on hamiltonian}
	\begin{aligned}
		H_N-\varepsilon \frac{n_+}{R^2}&\geq \frac{1}{\lfloor\frac{L}{2R}\rfloor+1}\sum_{n=0}^{\lfloor\frac{L}{2R}\rfloor} \left(\sum_{i=1}^N-\Delta_{L_n,i}-\varepsilon\frac{CQ_{L_n,i}}{L^2}+\sum_{i<j}v_{L_n}(x_i-x_j)\right)\\&=:\frac{1}{\lfloor\frac{L}{2R}\rfloor+1}\sum_{n=0}^{\lfloor\frac{L}{2R}\rfloor}H_{N,L_n}.
	\end{aligned}
\end{equation}
We choose $\varepsilon$ such that
\[\varepsilon C\leq \eta,\]
which allows us to apply Theorem~\ref{main bose theorem} on the local Hamiltonian.

The remainder of the proof follows a standard argument based on super additivity of the entropy, convexity of the energy functional, and positivity of the interaction; see \cite{LY} for the zero-temperature case and \cite{freeenergydilutebose,haberberger2024freeenergydilutebose} for positive temperature. We briefly outline the argument for completeness.

We denote the free energy of $H_{N,L_k}$ as $F(N,L_{n})$ and the free energy of a box of size $L_n$ with $m$ particles as $F_{L_n}(m)$.  From the super additivity of the entropy we find for any $\mu>0$,
\begin{equation}\label{free energy something}
F(N,L_n)-\mu N\geq \inf_{\sum_i m_i=N}\sum_{i}F_{L_n}(m_i)-\mu m_i.
\end{equation}
At low density, Theorem~\ref{main bose theorem} implies that $F_{L_n}(m)$ is positive and essentially quadratic in $m$. By further subdividing particles if necessary to maintain low density, we obtain
\begin{equation*}
	F_{L_n}(m_i)-\mu m_i\geq 0\qquad m_i\geq 3\frac{\mu}{4\pi} L_n^3
\end{equation*}
where in the above we use that $\mu$ will be chosen proportionally to $\rho$. We can then throw away the positive terms in the sum \cref{free energy something}, and insert \cref{main bose equation} to find
\begin{align*}
	F(N,L_n)-\mu N\geq \inf_{\sum_i m_i=N}\sum_{i,m_i\leq 3\frac{\mu}{4\pi} L_n^3} &4\pi a\frac{m_i^2}{L_n^3} (1+\frac{128}{15\sqrt{\pi}}\sqrt{\frac{m_i}{L_n^3} a^3}-C(\frac{m_i}{L_n^3} a^3)^{\frac{1}{2}+\eta})-\mu m_i\nonumber\\&+L_n^3\frac{T^{\frac{5}{2}}}{(2\pi)^3}\int_{\mathbb{R}^3}\log (1-e^{-\sqrt{p^4+\frac{16\pi p^2 \frac{m_i}{L_n^3}a}{T}}})dp.
\end{align*}
The right-hand side is convex in $m_i$, and choosing $\mu$ appropriately ensures that the minimum is attained at 
$m_i=\rho L_n^3$. This yields $\mu\sim 8\pi \rho(1+O(\sqrt{\rho a^3}))$. The infimum is thus achieved with each box having exactly $\rho L_n^3$ particles and combining the above with \cref{bound on hamiltonian} we find
\begin{align*}
		-T \log \Tr e^{-\beta (H_N-\eta\frac{n_+}{R^2})}&\geq 4\pi \rho a N(1+\frac{128}{15\sqrt{\pi}}\sqrt{\rho a^3}-C(\rho a^3)^{\frac{1}{2}+\eta})\nonumber\\&+R^3\frac{T^{\frac{5}{2}}}{(2\pi)^3}\int_{\mathbb{R}^3}\log (1-e^{-\sqrt{p^4+\frac{16\pi p^2 \rho a}{T}}})dp.  
\end{align*}
As before, combining this lower bound with the corresponding upper bound \cite{HABERBERGER2026110825} and the variational principle yields the desired result.
\end{proof}
\textit{Acknowledgements}: This work was funded by the Deutsche Forschungsgemeinschaft (DFG,
German Research Foundation) – Project-ID 470903074 – TRR 352. The author is grateful for the motivation and valuable discussions with Arnaud Triay and Søren Fournais. 
\printbibliography
\vspace{1em}
\small
\noindent
Deparment of mathematics, Ludwig-Maximilians-Universität-München, Germany \\
\textit{Email:} lukas.junge@math.lmu.de
\end{document}